\documentclass[conference]{IEEEtran}
\usepackage{amsmath}
\usepackage[mathscr]{eucal}
\usepackage{amssymb}
\usepackage{algorithmic}
\usepackage{algorithm}
\usepackage{float}
\usepackage{xspace}
\usepackage{graphicx}
\usepackage{epstopdf}
\usepackage{multirow,rotating}  
\usepackage{array}
\usepackage{comment}

 \newtheorem{thm}{Theorem}[section]
       
       \newtheorem{proposition}[thm]{Proposition}
       \newtheorem{cor}[thm]{Corollary}
    \newtheorem{mydef}[thm]{Definition}

\renewcommand{\IEEEQED}{\IEEEQEDopen}
       
\newcommand{\argmax}{\operatornamewithlimits{argmax}}
\newcommand{\argmin}{\operatornamewithlimits{argmin}}

\title{Minimum Description Length Principle for Maximum Entropy Model Selection}
\author{\IEEEauthorblockN{Gaurav Pandey and Ambedkar Dukkipati}
\IEEEauthorblockA{Department of Computer Science and Automation\\
Indian Institute of Science, Bangalore 560012 \\
Email: \{gaurav.pandey,ad\}.csa.iisc.ernet.in}}
\begin{document}
\maketitle

\begin{abstract}
In this paper, we treat the problem of selecting a maximum entropy
model given various feature subsets and their moments, as a model
selection problem, and present a minimum description length (MDL)
formulation to solve this problem. For this, we derive normalized
maximum likelihood (NML) codelength for these models. Furthermore, we
show that the minimax entropy method is a special case of maximum
entropy model selection, where one assumes that complexity of all the
models are equal. We extend our approach to discriminative maximum
entropy models. We apply our approach to gene selection problem to
select the number of moments for each gene for fixing the model. 
\end{abstract}

\section{Introduction}
 Model selection is central to statistics, and the most popular
 statistical techniques for 
 model selection are Akaike information criterion
 (AIC)~\cite{akaike1974new}, Bayesian information criterion
 (BIC)~\cite{schwarz1978estimating} and minimum description length
 (MDL)~\cite{rissanen1978modeling}.
 The basic idea behind MDL principle
 is to equate compression with finding regularities in the data. Since
 learning often involves finding regularities in the data, hence
 learning can be equated with compression as well. Hence, in MDL, we
 try to find the model that yields the maximum compression for the
 given observations.  

 The first MDL code introduced was the two-part code
 \cite{rissanen1978modeling}. It was shown that the the two-part code
 generalizes maximum entropy
 principle~\cite{feder1986maximum}. However, in the past two decades,
 the normalized maximum likelihood (NML) code, which is a version of
 MDL, has gained popularity among statisticians. This is particularly
 because of its minimax properties, as stated in
 \cite{shtar1987universal}, which the earlier versions of MDL did not
 possess. Efficient methods for computing NML codelengths for mixture
 models have been proposed in \cite{kont,hirai2011efficient}. In both these papers, the aim of model
 selection is to decide the optimum number of clusters in a clustering
 problem. 

  
  In a maximum entropy
  approach to density estimation one has to decide,
  \textit{a priori},  on the amount of `information' (for example,
  number of moments) that should be used from the data to fix the
  model. The minimax entropy principle \cite{zhu1997minimax} states
  that for given sets of features for the data, one should choose the
  set that minimizes the maximum entropy. However, it can easily be
  shown that if there are two feature subsets $\Phi_1$ and $\Phi_2$
  and $\Phi_1 \subset \Phi_2$, the minimax entropy principle will
  always prefer $\Phi_2$ over $\Phi_1$. Hence, though the minimax
  entropy principle is a good technique for choosing among sets of
  features with same cardinality, it cannot decide when the sets of
  features have varying cardinality. 

  In this paper, we study the NML codelength  of maximum
  entropy models. Towards this end, we formulate the problem of
  selecting a maximum entropy model given various feature subsets and
  their moments, as a model selection problem. We derive NML
  codelength for maximum entropy models and show that our approach is a
  generalization of minimax entropy principle. We also compute the NML
  codelength for discriminative maximum entropy models. 
   We apply our approach to gene selection problem for Leukemia data
  set and compare it with minimax entropy method.

  
\section{Preliminaries and Background} \label{sec:background}

  Let $\mathcal{M}$ be a family of probability distributions on the
  sample space $\mathcal{X}$. $\mathcal{X}^n$ denotes the sample
   space of all data samples of size $n$. $\mathbf{x}^n$ denotes an
   element in $\mathcal{X}^n$, where $\mathbf{x} =
   \left(x_1,...,x_d\right)$ is a vector in $\mathcal{X}$.  A two-part
   code encodes the data sample
   $D=\left(\mathbf{x}_1,\mathbf{x}_2,...,\mathbf{x}_n\right) \in
   \mathcal{X}^n$ by first encoding a distribution $p \in 
  \mathcal{M}$ and then the data $D$. The best hypothesis to explain data $D$ is then the one that
  minimizes the sum $L(p)+L(D|p)$ \cite{grunwald2004tutorial}. Thus,
  according to MDL principle 
  \begin{equation}
  L(D)=\underset{p \in \mathcal{M}}{\operatorname{min}} \left[ L(p)+L(D|p)\right] \label{crude} ,
  \end{equation}
  where $L(D)$ is the codelength of the data, $L(p)$ is the codelength of distribution and $L(D|p)$ is the codelength of data given the distribution. 
  $L(D|p)$ is a measure of the error of data $D$ with respect to
  distribution $p$. Hence, if $p$ approximates $D$ well enough,
  $L(D|p)$ should be small and vice versa. By Kraft's inequality,
  there exists a codelength function $L_p$ on $\mathcal{X}^n$ given by 
  $ L_p(D)=-\log(p(D))$.
 
  We an use $L_p(D)$ as
  $L(D|p)$ directly, since it is the unique minimizer of expected
  codelength, when $p$ is indeed the true distribution.
  The normalized maximum likelihood code is one of the several ways
  for constructing $L(p)$.
  To formalize the definition of normalized maximum likelihood, we
  need the following notion of `regret' \cite{grunwald2004tutorial}.  
  \begin{mydef} Let $\mathcal{M}$ be a model on $\mathcal{X}^n$ and
    let $\bar{p}$ be a probability distribution on
    $\mathcal{X}^n$. The regret of $\bar{p}$ with respect to
    $\mathcal{M}$ for data sample $\mathbf{x}^n$ is defined as  
  \begin{equation*} 
  \mathcal{R}(\bar{p})=-\log\bar{p}(\mathbf{x}^n)-\underset{p \in
    \mathcal{M}}{\operatorname{inf}}\left[-\log
    p(\mathbf{x}^n)\right]  \enspace . 
  \end{equation*}
  \end{mydef}
  The regret is nothing but the extra number of bits needed in
  encoding $\mathbf{x}^n$ using $\bar{p}$ instead of the optimal
  distribution for $\mathbf{x}^n$ in $\mathcal{M}$. The worst case
  regret denoted by $\mathcal{R}_{max}$ is defined as the maximum
  regret over all sequences in $\mathcal{X}^n$ 
  \begin{equation} \notag
  \mathcal{R}_{max}(\bar{p})=\underset{\mathbf{x}^n \in \mathcal{X}^n}{\operatorname{sup}}\left[-\log(\bar{p}(\mathbf{x}^n))-\underset{p \in \mathcal{M}}{\operatorname{inf}}[{-\log p(\mathbf{x}^n)}]\right] \enspace .
  \end{equation}
  Our aim is to find the distribution $\bar{p}$ that minimizes the maximum regret. 
  To this end, we define the complexity of a model
  $\mathrm{COMP}(\mathcal{M})$ as 
  \begin{equation}
  \mathrm{COMP}(\mathcal{M})=\log\int_{\mathbf{x}^n \in \mathcal{X}^n} p(\mathbf{x}^n;\hat{\theta}(\mathbf{x}^n))\: \mathrm{d}\mathbf{x}^n \label{comp} \enspace ,
  \end{equation}
where $\hat{\theta}(x^n)$ denote the maximum likelihood estimate (MLE)
of the parameter $\theta$ for the model $\mathcal{M}$ for the data
sample $x^n$. In the above equation and subsequent sections, the
integral is defined subject to existence. Also, we assume that MLE is
well defined for the model $\mathcal{M}$. 
  Furthermore, the error of a model is defined as 
  \begin{equation}
   \mathrm{ERR}(\mathcal{M},\mathbf{x}^n) = \inf_{p \in \mathcal{M}} \left[ -\log(p(\mathbf{x}^n) \right] \enspace . \label{eq:error_def}
  \end{equation}
  The following result is due to \cite{shtar1987universal}.
  \begin{thm}
  If the complexity of a model is  finite, then the minimax regret is
  uniquely achieved by the normalized maximum likelihood distribution
  given by 
  \begin{equation} \notag
  \bar{p}_{nml}(\mathbf{x}^n)=\frac{p(\mathbf{x}^n;\hat{\theta}(\mathbf{x}^n))} {\int_{\mathbf{y}^n \in \mathcal{X}^n} p(\mathbf{y}^n;\hat{\theta}(\mathbf{y}^n)) \: \mathrm{d}\mathbf{y}^n} \enspace .
  \end{equation}
  \end{thm}
  The corresponding codelength $-\log(p_{nml}(\mathbf{x}^N))$ also
  known as the stochastic complexity of the data sample $\mathbf{x}^n$
  is given by 
  \begin{equation*}
   \mathrm{NML}(\mathcal{M},\mathbf{x}^n) = \mathrm{ERR}(\mathcal{M}, \mathbf{x}^n) + \mathrm{COMP}(\mathcal{M}) \enspace .
  \end{equation*}

\section{Maximum Entropy Model Selection using MDL} \label{sec:maxent}
  \subsection{Problem Definition}
  Let $X$ be a random variable taking values in $\mathcal{X}$.
  Let $\Phi =
   \left\{\phi_1,...,\phi_m\right\}$ be a set of functions of $X$.
  The resultant linear
   family $\mathcal{L}_{\Phi, \mathbf{x}^n}$ is given by the set of all
   probability distributions that satisfy the constraints  
   \begin{equation}
     \int_{\mathbf{x} \in \mathcal{X}} p(\mathbf{x})\phi_k(\mathbf{x}) \: \mathrm{d}\mathbf{x} = \bar{\phi}_k(\mathbf{x}^n), \:\:\:\: 1 \le k \le m \enspace , \label{eq:linear}
   \end{equation}
   where $\bar{\phi}(\mathbf{x}^n)$ is the empirical estimate of $\phi$
   for the data $\mathbf{x}^n$. 
  The resulting maximum entropy model
  $\mathcal{M}_{\Phi_l}$ contains $p \in \Delta(\mathcal{X})$ such that
  \begin{equation}
   p(\mathbf{x}) = \exp\left(-\lambda_{0,l} - \sum_{k=1}^{m_l} \lambda_{k,l} \phi_{k,l}(\mathbf{x})\right) \label{eq:maxent} \enspace ,
  \end{equation}
    where $\Lambda = \left( \lambda_{0,l},...,\lambda_{m_l,l}\right) \in \mathbb{R}^{m_l+1}$. Here, $\lambda_{0,l}$ is the normalizing constant. 

  Given a set of maximum entropy models characterized by their
  function set $\Phi_l, 1 \le l \le r$, we use NML code to choose the
  model that best describes the data.

  \subsection{NML Codelength}
     
  The NML codelength of data $\mathbf{x}^n$ for a given model
  $\mathcal{M}$ is composed of two parts: (i) the error codelength and
  (ii) the complexity of the model. 
  \begin{proposition}
  Error codelength of data sequence $\mathbf{x}^n$ for the maximum entropy model $\mathcal{M}_\Phi$ is n times the maximum entropy of the corresponding linear family $\mathcal{L}_{\Phi,\mathbf{x}^n}$
  \begin{equation}
  \label{eq:error1}
  \mathrm{ERR}(\mathcal{M}_{\Phi},\mathbf{x}^n) = n H(p^*_{\mathbf{x}^n}),
  \end{equation}
  where $p^*_{\mathbf{x}^n}$ is the maximum entropy distribution of
  $\mathcal{L}_{\Phi, \mathbf{x}^n}$ given by \eqref{eq:linear}. 
\end{proposition}
  \begin{proof}
  First, we compute the error codelength of the data  $\mathbf{x}^n$
  for the model $\mathcal{M}_{\Phi}$. By definition,  
  \begin{equation*}
   \mathrm{ERR}(\mathcal{M}_{\Phi},\mathbf{x}^n) = \inf_{p \in \mathcal{M}_{\Phi}} \left\{ -\log(p(\mathbf{x}^n)) \right\}  \enspace .
  \end{equation*}
 Using definition of $\mathcal{M}_{\Phi}$ from equation \eqref{eq:maxent}, we get
  \begin{align}
   \mathrm{ERR}(\mathcal{M}_{\Phi} & ,\mathbf{x}^n) \notag \\
			 =& \inf_{\Lambda} \left[ -\sum_{i=1}^n \left(-\lambda_0 -\sum_{k=1}^m \lambda_k \phi_k(\mathbf{x}^{(i)}) \right) \right]\notag \\
			 =& \inf_{\Lambda} \left[ n\lambda_0 + \sum_{k=1}^m \lambda_k \sum_{i=1}^n \phi_k(\mathbf{x}^{(i)}) \right]\notag \\
			 =& \inf_{\Lambda}\left[  n\lambda_0 + \sum_{k=1}^m \lambda_k \left(n\bar{\phi}_k(\mathbf{x}^n)\right) \right] \notag \\
			 =& n\left[\inf_{\Lambda} \left( \lambda_0 + \sum_{k=1}^m \lambda_k \bar{\phi}_k(\mathbf{x}^n)\right) \right] \label{eq:Error_lambda} \enspace ,
  \end{align}
  where  $\bar{\phi}(\mathbf{x}^n)$ is the sample estimate of $\phi$ for the data  $\mathbf{x}^n$ and $\Lambda = (\lambda_0 ,..., \lambda_m) \in \mathbb{R}^{m+1}$. 

  Using Lagrange multipliers, it is easy to see that the maximum entropy distribution for the linear family $\mathcal{L}_{\Phi,\mathbf{x}^n}$ has the form
  \begin{equation*}
   p^*(\mathbf{x})=\exp\left(-\lambda_0^*- \sum_{k=1}^m {\lambda_k^*\phi_k(\mathbf{x})}\right) \: \mathrm{for \: all} \: \mathbf{x}\in \mathcal{X} \enspace ,
  \end{equation*}
  for some $\Lambda^* = (\lambda_0^* ,..., \lambda_m^*)$.
  Since maximum entropy distribution always exists for a linear family, the parameters $\Lambda^*$ can be obtained by maximizing the log likelihood function.
  \begin{align}
   \Lambda^* 	=&\argmax_{\Lambda} \left[ \sum_{i=1}^n \left(-\lambda_0 -\sum_{k=1}^m \lambda_k \phi_k(\mathbf{x}^{(i)}) \right)\right] \notag  \\
		=& \argmin_{\Lambda} \left[ \lambda_0 + \sum_{k=1}^m \lambda_k \bar{\phi}_k(\mathbf{x}^n) \right]  \label{eq:Lambda_star} \enspace .
  \end{align}
  where we remove the negative sign to change $\argmax$ to $\argmin$. The notation $\bar{\phi}_k(\mathbf{x}^n)$ is used to denote the empirical mean of $\bar{\phi_k}$ for the data  $\mathbf{x}^n$.

  The corresponding entropy is given by 
  \begin{align}
   H(p^*_{\mathbf{x}^n}) =& -\sum_{\mathbf{x}\in \mathcal{X}} p^*_{\mathbf{x}^n}(\mathbf{x}^{(i)}) \log(p^*_{\mathbf{x}^n}(\mathbf{x}^{(i)})) \notag \\
		=&  \lambda_0^* +\sum_{k=1}^m \lambda_k^* \sum_{\mathbf{x}\in \mathcal{X}} p^*_{\mathbf{x}^n}(\mathbf{x}^{(i)}) \phi_k(\mathbf{x}^{(i)}) \notag \\
		=&  \lambda_0^* +\sum_{k=1}^m \lambda_k^* \bar{\phi}_k(\mathbf{x}^{(i)}) \label{eq:entropy_and_lambda} \enspace ,
  \end{align}
  where the last equality follows from the definition of $\mathcal{L}_{\Phi,\mathbf{x}^n}$ in equation \eqref{eq:linear}. By combining equations \eqref{eq:Lambda_star} and \eqref{eq:entropy_and_lambda} and using the fact that maximum entropy distribution always exists for a linear family, we get
    \begin{equation}
  H(p^*_{\mathbf{x}^n})  = \min_{\Lambda} \left[ \lambda_0 + \sum_{k=1}^m \lambda_k \bar{\phi}_k(\mathbf{x}^n) \right]   \label{eq:entropy_lambda}  \enspace .
  \end{equation}
  By combining equations \eqref{eq:entropy_lambda} and \eqref{eq:Error_lambda}, we get 
  \begin{equation*}
   \mathrm{ERR}(\mathcal{M}_{\Phi}  ,\mathbf{x}^n) = nH(p^*_{\mathbf{x}^n}) \enspace .
  \end{equation*}
  For fixed $\mathbf{x}^n$, the error depends on the function set $\Phi$ through the above equation. As the no. of functions in the function set is increased, the size of the linear family decreases. Hence, the entropy of the maximum entropy distribution also decreases, since we are restricted to search for the maximum entropy distribution in a smaller space. Hence, error of the model decreases.
  \end{proof}

\begin{cor}
 Complexity of the maximum entropy model $\mathcal{M}_\Phi$ is given by
\begin{equation}
\label{eq:comp1}
 \mathrm{COMP}(\mathcal{M}_\Phi) = \log{\int_{\mathbf{y}^n \in \mathcal{X}^n} {\exp(-n H(p^*_{\mathbf{y}^n})) }}\:\mathrm{d}\mathbf{y}^n \enspace ,
\end{equation}
where $p^*_{\mathbf{y}^n}$ is the maximum entropy distribution of
$\mathcal{L}_{\Phi, \mathbf{y}^n}$.
\end{cor}
\begin{proof}
We have
  \begin{equation}
   \max_{p \in \mathcal{M}_{\Phi}} p(\mathbf{y}^n) = \exp( \max_{p \in \mathcal{M}_{\Phi}} \log p(\mathbf{y}^n))   
					   = \exp( -nH(p^*(\mathbf{x}^n)))    \label{eq:comp_entropy} \enspace ,
  \end{equation}
  where we have used the definition of  error in \eqref{eq:error_def} and its relationship with entropy in \eqref{eq:error1} to get the result. By similar arguments as above, it is easy to see that the complexity of the model increases by increase in the number of constraints.
By using \eqref{eq:comp_entropy} and \eqref{eq:comp1} we get the desired result.
\end{proof}
Hence, the NML codelength (also known as stochastic complexity) of $\mathbf{x}^n$ for the model $M_{\Phi}$ is given by
\begin{equation*}
 \mathrm{NML}(\mathcal{M}_\Phi,\mathbf{x}^n) =  n H(p^*_{\mathbf{x}^n}) + \log{\int_{\mathbf{y}^n \in \mathcal{X}^n} \!\!\!\!\!\!\!\!\!\!\!\!\!\! {\exp(-n H(p^*_{\mathbf{y}^n})) d\mathbf{y}^n }}\enspace .
\end{equation*}

\subsection{Generalization Of Minimax Entropy Principle}  

In this section, we show that the presented NML formulation for maximum entropy is a generalization of the minimax entropy principle \cite{zhu1997minimax}, where this principle has been used for feature selection in texture modeling.

Let $\Phi_1,...,\Phi_l$ be sets of functions from $\mathcal{X}$ to the set of real numbers. Corresponding to each set $\Phi_p$, there exists a maximum entropy model $\mathcal{M}_{\Phi_p}$ and vice-versa. 
The MDL principle states that given a set of models for the data, one should choose the model that minimizes the codelength of the data. Here, the codelength that we are interested in is the NML codelength (also known as stochastic complexity). Since, there exists a one-one relationship between the maximum entropy models and the function sets $\Phi_p$, the model selection problem can be reframed as
\begin{equation}
  \hat{\Phi} = \argmin_{\Phi_k, 1 \le k \le p} \left[n H(p^*_{\mathbf{x}^n,k})  \vphantom{\int_x} \right. \nonumber 
	      + \left. \log\!\!{\int_{\mathbf{y}^n \in \mathcal{X}^n}\!\!\!\!\!\!\!\!\!\!\!\!\!\!\! {\exp(-n H(p^*_{\mathbf{y}^n,k})) \:\mathrm{d}\mathbf{y}^n }}\right] \enspace .
\end{equation}
If we assume that all our models have the same complexity, then the second term in R.H.S can be ignored. Since n, the size of data  is a constant, the model selection problem becomes
\begin{align*}
  \hat{\Phi} =& \argmin_{\Phi_k, 1 \le k \le p}  H(p^*_{\mathbf{x}^n,k})   = \argmin_{\Phi_k, 1 \le k \le p} \ \max_{p \in \mathcal{L}_{\Phi_k, \mathbf{x}^n}} \  H(p) \enspace .
\end{align*}
This is the classical minimax entropy principle given in \cite{zhu1997minimax}. Hence, the minimax entropy principle is a special case of the MDL principle where the complexity of all the models are assumed to be the same and the models assumed are the maximum entropy models. 

\section{Discriminative Models} \label{sec:discriminative}
Discriminative methods for classification, model the conditional
probability distribution $p(c|\mathbf{x})$, where $c$ is the class
label for data $x$.  
Maximum entropy based discriminative classification tries to find the
probability distribution with the maximum conditional entropy $H(C|X)$
subject to some constraints \cite{berger1996maximum}, where $C$ is the
class variable. Initially, information is extracted from the data in
the form of empirical means of functions. These empirical values are
then equated to their expected values, thereby forming a set of
constraints. The classification model is constructed by finding the
maximum entropy distribution subject to these sets of constraints. We
use MDL to decide the amount of information to extract from the data
in the form of functions of features. A straightforward application of
this technique is feature selection. 

The maximum entropy discriminative model $\mathcal{M}_\Phi$, where $\Phi = \left\{\phi_1,...,\phi_m\right\}$ is the set of all probability distributions of the form
\begin{equation*}
 p(c|\mathbf{x}) = \frac{\exp(-\sum_{k = 1}^m \lambda_k \phi_k(\mathbf{x}, c))}{\sum_{c \in \mathcal{C}}{\exp(-\sum_{k = 1}^m \lambda_k \phi_k(\mathbf{x}, c))}} \enspace ,
\end{equation*}
Let us denote the denominator in above equation as
$Z(\mathbf{x})$.

Since, we are not interested in modelling $p(\mathbf{x})$, we use the
empirical distribution $\tilde{p}(\mathbf{x})$to approximate
$p(\mathbf{x})$~\cite{berger1996maximum}.
The empirical
distribution $\tilde{p}(\mathbf{x})$ is given by
$\tilde{p}(\mathbf{x}) = \frac{1}{n},\:\:  \mathbf{x}\in \mathbf{x}^n$ and
$\tilde{p}(\mathbf{x}) =0$ otherwise.
Hence, the constraints become
\begin{align}
 \sum_{i=1}^n \sum_{c \in \mathcal{C}} p(c|\mathbf{x}^{(i)}) \phi_k(\mathbf{x}^{(i)}, c) &= \sum_{i=1}^n \phi_k(\mathbf{x}^{(i)}, c^{(i)}) \enspace . \label{eq:linearfamily} 
\end{align}

As discussed in \cite{tabus2003classification}, the sender-receiver
model assumed here is as follows. Both sender and receiver have the
data $\mathbf{x}^n$. The sender is interested in sending the class
labels $c^n$. If he sends the class labels without compression, he
needs to send $n\log(|C|)$ bits. If, however, he uses the data to
compute a probability distribution over the class labels, and then
compress $c^n$ using that distribution, he may get a shorter
codelength for $c^n$. His goal is to minimize this codelength, such
that the receiver can recover the class labels from this code.  
\begin{proposition}
 The Error codelength of $c^n|\mathbf{x}^n$ for the conditional model $\mathcal{M}_\Phi$ is equal to n times the maximum conditional entropy of the model $\mathcal{L}_{\Phi, \mathbf{x}^n, c^n}$.
\end{proposition}
  \begin{proof}
    Error of the conditional model is given by
    \begin{align}
     &\mathrm{ERR}(\mathcal{M}_\Phi, c^n|\mathbf{x}^n) 	= \inf_{p \in \mathcal{M}_\Phi} -\log{p(c^n|\mathbf{x}^n)} \notag \\
				=& \inf_\Lambda \left[ \sum_{k=1}^m \lambda_k \sum_{i=1}^n \phi_k(\mathbf{x}^{(i)}, c^{(i)})  + \sum_{i=1}^n \log{Z_\Lambda(\mathbf{x}^{(i)})} \right] \label{eq:error_lambda_disc}\!\!\! \enspace .
    \end{align}
 where we have used similar reasoning as in  \eqref{eq:Error_lambda} to get the last statement. 
  Also, the maximum conditional entropy distribution can be obtained by maximizing the corresponding log-likelihood function.
  Hence,
     $p^* = \argmax_{p \in  \mathcal{M}_\Phi} \log{p(c^n|\mathbf{x}^n)} \:\:\: .$
  Correspondingly,
  \begin{equation}
    \Lambda^*\!\!  = \!\!\argmin_\Lambda\!\! \sum_{k=1}^m \!\!\lambda_k \sum_{i=1}^n \phi_k(\mathbf{x}^{(i)}, c^{(i)}) + \sum_{i=1}^n \log{Z_\Lambda(\mathbf{x}^{(i)})} \label{eq:Lambda_star_disc}\!\!\! \enspace .
  \end{equation}

  The corresponding conditional entropy is given by
  \begin{align}
   H(p^*) =& -\int_{\mathbf{x}\in \mathcal{X}} \sum_{c \in \mathcal{C}} p^*(c |\mathbf{x})\tilde{p}(\mathbf{x}) \log p^*(c | x) \notag\\
	  =& \frac{1}{n} \left[ \sum_{k=1}^m \lambda_k^* \left(\sum_{i=1}^n \sum_{c \in \mathcal{C}} p^*(c|\mathbf{x}^{(i)}) \phi_k(\mathbf{x}^{(i)}, c) \right) \right. \notag\\
	  &+ \left. \sum_{i=1}^n \sum_{c \in \mathcal{C}} p^*(c|\mathbf{x}^{(i)}) \log{Z_\Lambda(\mathbf{x}^{(i)})} \right] \nonumber \\
	  =& \frac{1}{n}\!\! \left[\sum_{k=1}^m \lambda_k^*\!\! \left(\sum_{i=1}^n \phi_k(\mathbf{x}^{(i)}, c^{(i)})\right) \right.  
	  \!\!\!+\!\!\! \left. \sum_{i=1}^n   \log{Z_\Lambda(\mathbf{x}^{(i)})}\right] \label{eq:entropy_lambda_disc} \!\!\!\enspace.
  \end{align}

 Here the first equality follows from the definition of conditional entropy as used in \cite{berger1996maximum}. We use the definition of $\tilde{p}$ to convert the integral to a summation. By using \eqref{eq:linearfamily} and the fact that $p(c|\mathbf{x}^{(i)})$ must sum up to 1, we obtain the fourth equality.
 
 Using \eqref{eq:Lambda_star_disc} and \eqref{eq:entropy_lambda_disc}, we obtain
  \begin{equation}
   H(p^*)\! =\! \frac{1}{n} \min_\Lambda \!\!\left[ \sum_{k=1}^m \lambda_k \sum_{i=1}^n \phi_k(\mathbf{x}^{(i)}, c^{(i)}) \right. 
    \left. + \sum_{i=1}^n \log{Z_\Lambda(\mathbf{x}^{(i)})} \right] \label{eq:entropy_lambda_disc1} \!\!\!\!\enspace .
  \end{equation}

 Replacing the above equation in \eqref{eq:error_lambda_disc}, we get the desired result. 
\end{proof}
  
\begin{cor}
 The complexity of the conditional model $\mathcal{M}_\Phi$ is given by
  \begin{equation*}
   \mathrm{COMP}(\mathcal{M}_\Phi) = \sum_{y^n \in \mathcal{C}^n} \exp(-nH(p^*_{y^n})) \enspace .
  \end{equation*}
\end{cor}


\section{Application to Gene Selection} \label{sec:gene_selection}
  We use gene selection as an example to illustrate discriminative
  model selection for maximum entropy models.  The dataset used is
  Leukemia dataset available publicly at
  http://www.genome.wi.mit.edu. The dataset was also used in
  \cite{tabus2003classification} to illustrate NML model selection for
  discrete regression. The data set consists of two classes: acute
  myeloid leukemia (AML) and acute lymphoblastic leukemia (ALL). There
  are 38 training samples and 34 independent test samples in the
  data. The data consists of 7129 genes. The genes are  preprocessed
  as recommended in \cite{dudoit2002comparison}.  

 Assuming the sender-receiver model discussed above, the sender needs
 38 bits or 26.34 nats in order to send the class labels of training
 data to receiver. If the NML code is used, the sender needs 24.99
 nats. Since the sender and receiver both contain the microarray data,
 the sender can use the microarray data to compress the class labels
 much more than can be obtained wihout the microarray
 data. Specifically, we are interested in finding the genes which
 gives the best compression, or the minimum NML codelength. 

\begin{figure} [ht]
  \centering
  \includegraphics[height=0.2\textwidth, width=0.45\textwidth]{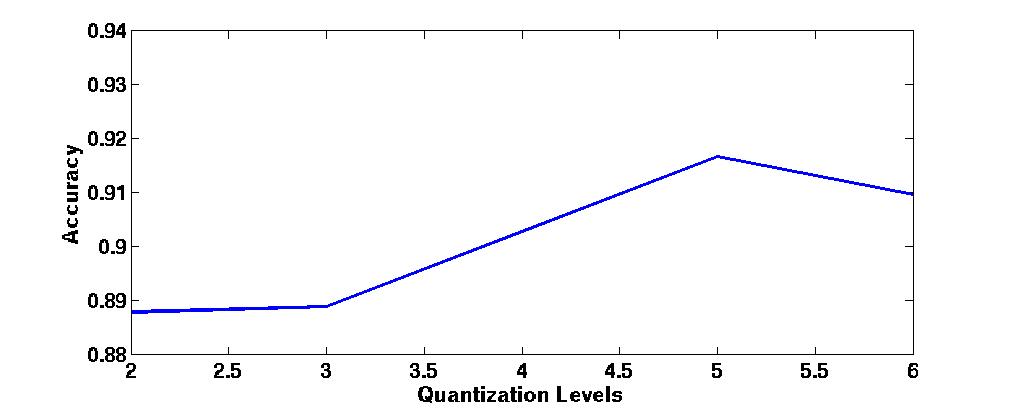}
  \caption{Change in accuracy with quantization level for the top $25$ genes }
  \label{accuracy}
  \end{figure}

\begin{figure} [ht]
  \centering
  \includegraphics[height=0.2\textwidth, width=0.45\textwidth]{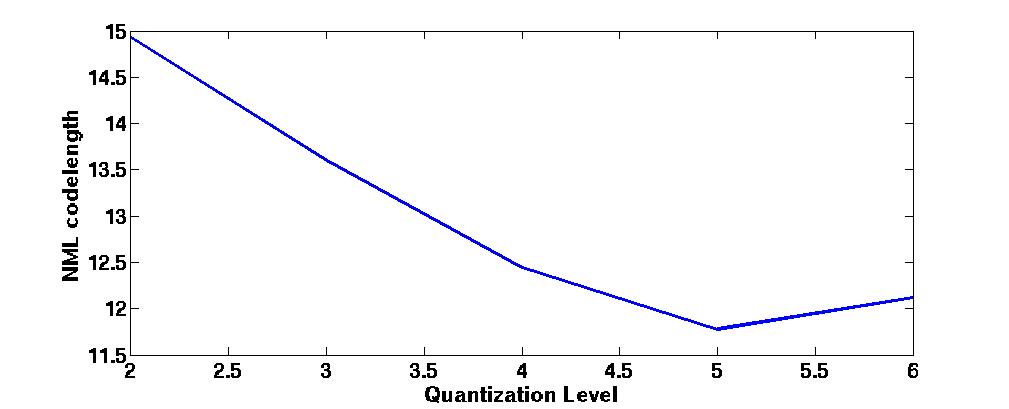}
  \caption{Change in average NML codelength with quantization level for the top $25$ genes}
  \label{nml_codelength}
  \end{figure}
  
  \begin{figure} [ht]
  \centering
  \includegraphics[height=0.2\textwidth, width=0.52\textwidth]{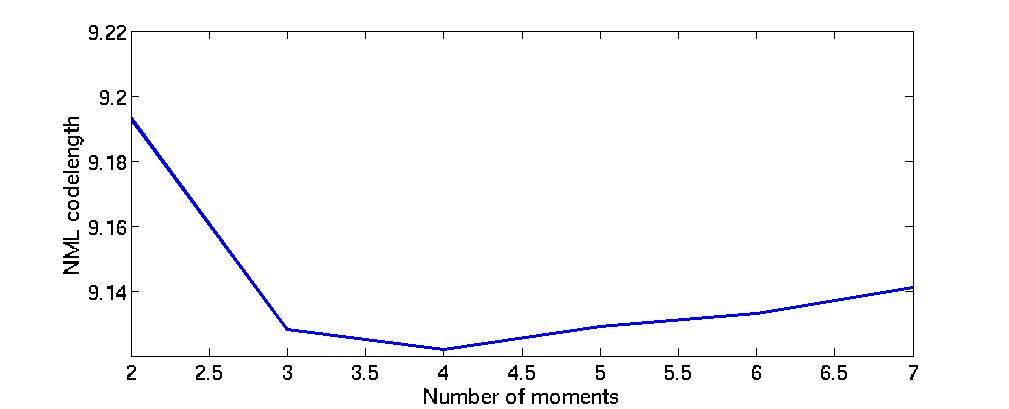}
  \caption{Variation of NML codelength of class labels for gene M16038 with the number of moment constraints m}
  \label{nml_vs_m}
  \end{figure}

 \begin{figure} [ht]
  \centering
  \includegraphics[height=0.3\textwidth, width=0.52\textwidth]{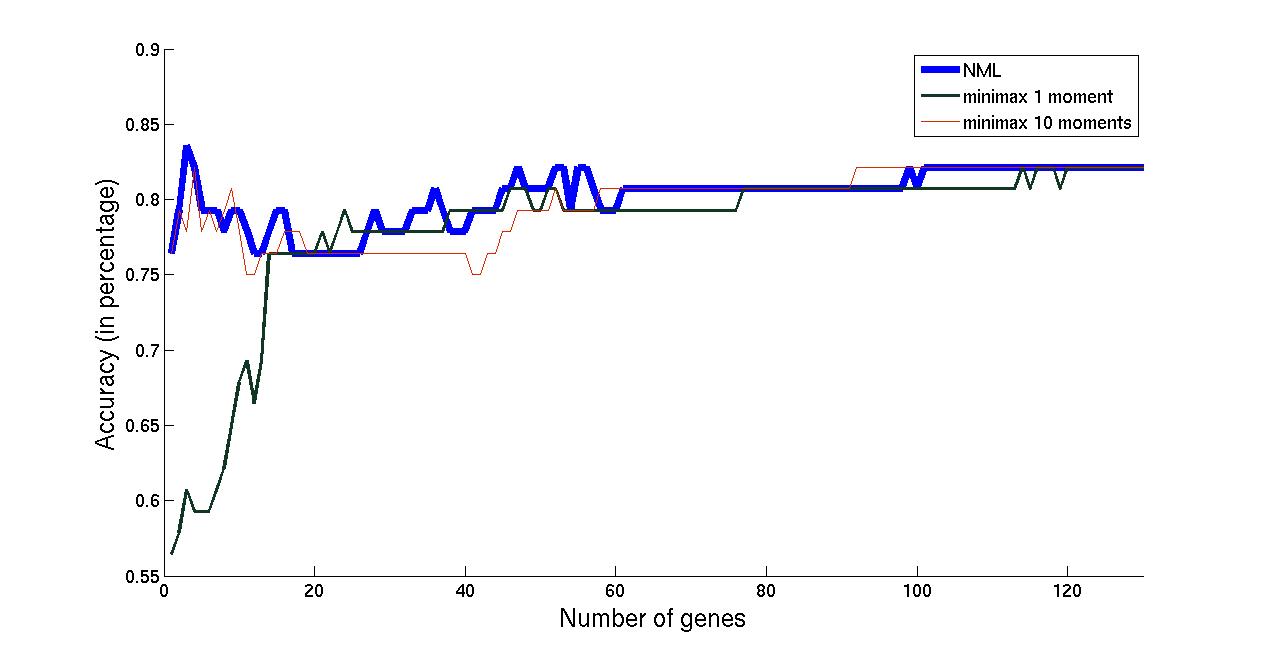}
  \caption{In this figure, we compare the accuracy of various  maximum entropy classifiers based on NML and minimax entropy. The number of genes in the maximum entropy classifier are varied from 1 to 130. Class conditional independence among the features is assumed. Minimax entropy does not help in deciding on the number of moments to use from the data while building a statistical model for the data. Hence, for the graphs `minimax 1 moment' and `minimax 10 moments', we use first moment and first ten moments of the data to build a statistical model. However, the proposed NML based maximum entropy classifier fixes number of moments by MDL principle.}
  \label{classifier}
  \end{figure}

For the purpose of our algorithm, we quantize the genes to various levels. We claim that quantizing a gene reduces the risk of overfitting of the model to data. To support our claim, we have also  plotted the change in accuracy with quantization level in Figure~\ref{accuracy} for the top 25 genes. As can be seen from the graph, increasing quantizaton level from $5$ to $6$ results in a decrease in accuracy. We have also plotted a graph for change in average NML codelength with quantization level for the top 25 genes in Figure~\ref{nml_codelength}. An interesting observation is that the minima of NML codelength coincides exactly with the maxima of accuracy. A similar trend was obsrved when the number of genes were changed.
   
 Hence, we quantize each gene to 5 levels. Other than the advantages of quantization mentioned above, quantization is also necessary for the current problem as the problem of calculating  complexity can become intractable even for moderate n. The constraints that we use are moment constraints, that is $\phi_k(\mathbf{x}) = \mathbf{x}^k, 1 \le k \le m $. We vary the value of m from 1 to 7 to get a sequence of maximum  entropy models. The NML codelength of the class labels is calculated for each such model. The model that results in the minimum NML codelength is selected for each gene. 

 It was observed that for most genes, the NML codelength decreased sharply when m was increased from 1 to 2. The change in values of NML codelength was less noticeable for $m \ge 2$. The variation of NML codelength of class labels for a typical gene are shown in figure~\ref{nml_vs_m}. In order to make the changes in NML codelength more visible, we skip the NML codelength for m=1.    

  Our approach for ranking genes is as follows. For each gene, we select the value of m that gives the minimum NML codelength. We then sort the genes in increasing order of their minimum NML codelengths. The minimum codelength achieved is 8.35 nats, which is much smaller than 24.99 nats achieved without using the microarray data. Since compression is equated with finding regularity according to Minimum Description Length principle, hence, it can be stated that the topmost gene is able to discover a lot of regularity in the data.

  Finally, we use MDL to build a classifier. The amount of information to use for each gene is decided, by using MDL to fix the number of moments. MDL is used to rank the features. Then, we use class conditional independence among features to build a maximum entropy classifier. The number of genes used for the classifier are varied from 1 to 130. The resultant graph is compared with other maximum entropy classifiers in Figure~\ref{classifier}, where the amount of information used per gene is the same for all genes.

\section{Conclusion}
Finding appropriate feature functions and the number of moments is
important to any maximum entropy method. In this paper, we pose this
problem as a model selection problem and develop an
MDL based method to solve this problem. We showed that this approach
generalizes minimax entropy principle of~\cite{zhu1997minimax}. We
derived NML codelength in this respect, and extended it to
discriminative maximum entropy model selection. We tested our proposed
method for gene selection problem to decide on the quantization level
and number of moments for each gene.
Finally, we selected the genes
based on the codelength of the class labels and compared the
simulation results with minimax entropy method. The bottleneck for
using MDL for model selection in discriminative classification is the
computation of complexity. More efficient approximations to calculate
the complexity need to be developed to employ this approach for
problems involving larger data sets.

\end{document}